\pdfoutput=1 %
\documentclass{easychair}

\usepackage{doc}
\usepackage{framed}
\usepackage{booktabs}
\usepackage{adjustbox}
\usepackage{amsmath}
\usepackage{amsfonts}
\usepackage{amsthm}
\usepackage{xspace}
\newtheorem{definition}{Definition}
\newtheorem{proposition}{Proposition}
\newtheorem{theorem}{Theorem}
\newtheorem{lemma}{Lemma}

\newtheorem{fact}{Fact}
\makeatletter
\newtheorem*{rep@theorem}{\rep@title}
\newcommand{\newreptheorem}[2]{%
  \newenvironment{rep#1}[1]{%
    \def\rep@title{#2~\ref{##1}}%
    \begin{rep@theorem}}%
    {\end{rep@theorem}}}
\makeatother
\newreptheorem{theorem}{Theorem}
\newreptheorem{proposition}{Proposition}
\newcommand{\Func}[1][\relax]{\ensuremath{\mathcal{F}_\mathsf{#1}}\xspace}

\usepackage{algorithm}
\usepackage[noend]{algpseudocode}
\usepackage{mathtools}
\usepackage{fancyvrb}
\usepackage{lineno}
\usepackage{xcolor}
\usepackage{pgfplots}
\usepackage{subcaption}
\usepackage[dvipsnames]{xcolor} 
\usepgfplotslibrary{groupplots}
\newcommand{\F}{\mathbb{F}}

\usepackage[normalem]{ulem} %

\newenvironment{nffunc}[1]{%
 \small
  \begin{framed}
  \begin{minipage}{\linewidth}
\vspace{-3pt}
       {\begin{center}\underline{\textbf{Functionality} #1}\end{center}}
}{%
\vspace{-3pt}
\end{minipage}
  \end{framed}
}

\newenvironment{nfprot}[1]{%
  \small
  \begin{framed}
    \begin{minipage}{\linewidth}
\vspace{-3pt}
      {\begin{center}\underline{\textbf{Protocol} #1}\end{center}}
}{%
\vspace{-3pt}
    \end{minipage}
  \end{framed}
}

\title{Scaling Zero Knowledge UNSAT Verification via Normalized Chaining}

\author{
    Ashwin Karthikeyan \inst{1}\thanks{These authors contributed equally to this work.}
\and
    Ethan Kharitonov\inst{2}\footnotemark[1]
\and
    Kuldeep S. Meel\inst{1,3}
\and 
    Anwar Hithnawi\inst{1}
}

\institute{
  University of Toronto, Canada\\
  \email{ashwin@cs.toronto.edu, meel@cs.toronto.edu, ahithnawi@cs.toronto.edu}
\and
   Independent, Canada\\
   \email{ethan.kharitonov@gmail.com}\\
\and
   Georgia Institute of Technology, USA\\
   \email{meel@gatech.edu}
 }

\authorrunning{Karthikeyan, Kharitonov, Meel and Hithnawi}

\titlerunning{Scaling Zero Knowledge UNSAT Verification}

\begin{document}

\maketitle

\begin{abstract}

Proofs of UNSAT are a standard primitive in formal verification and software assurance. In many real-world settings, the proof itself encodes proprietary or security-sensitive information, making public disclosure undesirable. Zero-knowledge certification of UNSAT addresses this tension: it enables a prover to convince a verifier that no satisfying assignment exists, without revealing anything about the underlying proof beyond its validity. Luo et al. recently introduced \textsc{ZkUnsat}, a protocol that achieves this goal by proving the validity of a weakened resolution proof in zero knowledge. \textsc{ZkUnsat} demonstrates the feasibility of zero-knowledge certification; however, its scalability to larger, real-world instances is constrained by substantial prover memory overhead, limiting its real-world applicability.
Motivated by advances in UNSAT proof formats such as LRAT, which enable efficient plain-text verification, we present a preprocessing technique that improves the efficiency of \textsc{ZkUnsat} without introducing additional leakage. Our approach normalizes the proof so that each derived clause is justified by a resolution chain of fixed public length $k$. This eliminates chain-length leakage and reduces prover memory usage. 
With $k = 16$, our method certifies roughly $62\%$ more instances than 
baseline \textsc{ZkUnsat} on the SAT 2002 competition benchmarks. Furthermore, 
for an equivalent number of certified instances, the memory footprint drops to 
under $25\%$ of that required by the baseline.

\end{abstract}

\section{Introduction}

As SAT solvers are increasingly deployed across diverse domains---from verifying safety-critical systems \cite{BCC+99} to 
solving long-standing mathematical conjectures \cite{HKM16, MLV+03}---the question of 
trust has become paramount. Modern SAT solvers are remarkably effective, 
capable of solving instances with millions of variables and clauses that 
would have been intractable mere decades ago \cite{K15}. Yet this effectiveness 
comes with a caveat: state-of-the-art solvers are large, highly optimized 
systems comprising hundreds of thousands of lines of code, making them 
susceptible to subtle implementation bugs \cite{BLB10}. 

To address these concerns, the SAT community has made substantial progress over the past two decades in proof-logging techniques. Early work centered on independently checkable resolution proofs \cite{ZM03}, which, while theoretically sound and complete, imposed significant overhead on both proof generation and verification. The introduction of the RUP (Reverse Unit Propagation) \cite{GN03,G08} proof format marked a turning point; it reduced the burden on solvers by allowing them to emit proofs without explicit resolution chains, and DRUP (Deletion Reverse Unit Propagation) \cite{HHW13DRUP} extended this approach with clause deletion information, making proofs efficiently checkable. More recently, LRAT (Linear Resolution Asymmetric Tautology)~\cite{CHH+17} pushed 
the boundary further by enabling more efficient, formally verified proof checking. 
These advances, coupled with highly optimized proof checkers~\cite{L24} and proof 
trimming tools such as \texttt{drat-trim}~\cite{WHH14} and 
\texttt{lrat-trim}~\cite{PFB23}, have made it practical for solvers to emit 
independently checkable certificates of unsatisfiability. This ecosystem of proof formats and checkers has established trust in solver outputs across diverse applications.

However, proof verification requires complete disclosure of the underlying proof, 
exposing the very witness the prover may wish to protect. This barrier 
increasingly prevents adoption in privacy-sensitive settings. Consider a proprietary SAT solver whose competitive advantage lies 
in carefully tuned heuristics---the proof itself may reveal the decision 
sequence and branching patterns that encode years of engineering effort. In 
industrial verification workflows, the unsatisfiability proof of a formula 
encoding a system's specification might leak sensitive architectural details 
or security properties that the vendor wishes to protect.

Zero-knowledge (ZK) proofs offer an elegant resolution to this tension. In a 
zero-knowledge proof system, a prover convinces a verifier of a statement's 
truth without revealing any information about the witness beyond its 
existence. The recent work of Luo et al.~\cite{LAH+22} introduced \textsc{ZkUnsat}, a 
protocol that achieves zero-knowledge certification of Boolean formula 
unsatisfiability. While it does not match the efficiency of plaintext proof 
checkers, \textsc{ZkUnsat} represents the first zero-knowledge \textsc{Unsat} 
protocol with practical viability, in contrast to earlier constructions that 
remained largely theoretical~\cite{BGG+90, LFKN92}. Its key insight is to 
prove the validity of a \textit{weakened resolution proof} in zero knowledge, 
using a polynomial commitment scheme and a modified zero-knowledge 
RAM~\cite{FKL+21} to hide the proof's contents while establishing its 
correctness. Despite this important feasibility demonstration, 
\textsc{ZkUnsat}'s scalability to real-world instances remains severely 
constrained: memory exhaustion has been identified as the primary 
limitation~\cite{LAH+22}, a finding we independently confirm on the SAT 2002 
competition benchmarks, where \textsc{ZkUnsat} exhausts available memory on 
the majority of instances (see Section \ref{sec:mem-bottleneck-from-unfold}). This gap underscores the need to reduce the 
overhead of zero-knowledge proof checking to a level practical for real-world 
use.

The central contribution of this work is to substantially improve the 
scalability of zero-knowledge \textsc{Unsat} certification, bringing it closer to practical deployment. To this end, we conduct a systematic 
analysis to identify the bottleneck constraining \textsc{ZkUnsat}'s 
performance. \textsc{ZkUnsat} expects proofs in which each proof line 
represents a single resolution step---every derived clause must be produced 
from exactly two parent clauses. To satisfy this requirement, proofs produced 
by modern SAT solvers must be \textit{fully unfolded}: a chained resolution 
step of length $m$ is decomposed into $m-1$ intermediate resolution steps, 
each introducing an intermediate clause. Our investigation reveals that this 
unfolding is the dominant source of memory overhead. Each intermediate clause 
requires its own index and cryptographic commitment in the zero-knowledge RAM, 
causing the prover's memory footprint to scale with the fully unfolded proof 
size rather than the original proof structure.

To quantify this overhead, we introduce a variant that verifies chained proofs 
directly without unfolding, which we call \mbox{Var-Chain-\textsc{ZkUnsat}}. On 510 
proofs from the SAT 2002 competition benchmarks, \textsc{ZkUnsat} successfully 
certifies 220 instances before exhausting the 32 GB memory limit, with 290 
instances aborted due to memory exhaustion. In contrast, 
\mbox{Var-Chain-\textsc{ZkUnsat}} certifies 281 instances, with only 103 
memory-related failures and peak memory usage remaining below 15 GB on most 
instances. This reduction confirms that unfolding is the primary bottleneck. 
Furthermore, \textsc{ZkUnsat} exhibits a tight correlation between runtime and 
peak memory---both scale together---whereas \mbox{Var-Chain-\textsc{ZkUnsat}} 
shows no such pattern, with memory usage entirely decoupled from execution 
time. However, bypassing 
unfolding introduces a critical issue: 
\mbox{Var-Chain-\textsc{ZkUnsat}} reveals the length of each resolution chain 
to the verifier, directly violating zero-knowledge guarantees.
Our key insight is that 
weakened resolution chains can be \textit{normalized} to a fixed public length 
$k$ while preserving zero-knowledge: the verifier learns only the number of 
lines in the normalized proof and the public parameter $k$, with no 
information about the structure of the underlying proof.

To achieve this, we present a preprocessing algorithm that transforms any chained 
resolution proof into one where every chain has length exactly $k$. The 
algorithm operates by padding shorter chains (exploiting the fact that a 
clause can be weakly resolved with itself) and decomposing longer chains into 
sequences of $k$-length chains. This normalization preserves zero-knowledge: 
the verifier learns only the number of proof lines and the public parameter 
$k$, with no information about the original chain-length distribution. Yet it 
captures the memory benefits of chaining by eliminating the intermediate 
clauses from unfolding. With $k=16$, our method---which we call 
$k$-Chain-\textsc{ZkUnsat}---reduces average peak prover memory and enables 
certification of approximately 62\% more instances compared to baseline 
\textsc{ZkUnsat} on the SAT 2002 benchmarks. Critically, this improvement shifts the performance bottleneck from memory constraints to computational runtime, 
substantially expanding the range of instances amenable to zero-knowledge 
verification. Beyond improving  scalability, we make a methodological contribution by 
developing an accurate estimator for the peak memory usage of the
certification protocol.

\textbf{Organization}: The remainder of this paper is organized as follows. Section~\ref{sec:prelims} introduces the necessary preliminaries, covering the logical foundations of resolution-based refutation and the cryptographic primitives underlying zero-knowledge certification. Section~\ref{sec:background} provides background on proof formats, and the \textsc{ZkUnsat} 
construction. Section~\ref{sec:observations} presents our analysis identifying unfolding as the 
primary bottleneck and introduces the normalization algorithm. Section~\ref{sec:Thm-and-privacy-preservation} establishes the zero-knowledge guarantees of the normalized protocol and derives a peak memory estimator. Section~\ref{sec:evaluation} 
details our experimental evaluation on the SAT 2002 benchmarks, demonstrating 
both the performance improvements and the accuracy of our estimators. We 
conclude in Section~\ref{sec:conclusion} with directions for future work.

\section{Preliminaries}\label{sec:prelims}

\subsection{Logic}

\paragraph{Boolean formulas.} Let $V$ be a finite set of Boolean variables 
and $\mathbb{B} = \{\top, \bot\}$. A Boolean formula $\varphi: \mathbb{B}^{V} 
\longrightarrow \mathbb{B}$ is \textit{satisfiable} (SAT) if there exists an 
assignment $\sigma \in \mathbb{B}^{V}$ such that $\varphi(\sigma) = \top$, 
and \textit{unsatisfiable} (UNSAT) otherwise. A literal is a variable $x \in 
V$ or its negation $\neg x$. A clause is a disjunction of literals $\ell_{1} 
\lor \dots \lor \ell_{n}$. A formula in conjunctive normal form (CNF) is a 
conjunction of clauses $\varphi = C_{1} \land \dots \land C_{m}$. We identify 
clauses with their sets of literals and formulas with their sets of clauses, 
writing $\ell \in C$ and $C \in \varphi$ accordingly.

\paragraph{Resolution proofs.} A resolution proof of unsatisfiability 
for a CNF formula $\varphi$ is a sequence of clauses terminating in the empty 
clause $\bot$, where each clause is either from $\varphi$ or derived by 
resolving two previous clauses. Given clauses $C = \ell_{1} \lor \dots \lor 
\ell_{n} \lor x$ and $C' = \ell'_{1} \lor \dots \lor \ell'_{m} \lor \neg x$ 
for some variable $x \in V$, their resolvent on $x$ is $C_{\text{res}} = 
\ell_{1} \lor \dots \lor \ell_{n} \lor \ell'_{1} \lor \dots \lor \ell'_{m}$.

A \textit{chained resolution} step of length $k$ specifies clauses $C_1, 
\dots, C_{k+1}$ and pivot variables $x_1, \dots, x_k$. The chain is evaluated 
left-to-right: $C_1$ resolves with $C_2$ on $x_1$ to produce an intermediate 
clause, which then resolves with $C_3$ on $x_2$, and so forth until the final 
resolvent is obtained. Modern proof formats such as DRAT \cite{WHH14} and LRAT 
\cite{CHH+17} employ chained resolution to reduce proof size by avoiding explicit 
representation of intermediate clauses.

\paragraph{Weakened resolution.} A clause $D'$ is a weakening of clause $D$ if $D' = D \lor S$ for some clause $S$. A clause $C_{\text{wres}}$ is called a \textit{weakened resolvent} of clauses $C_1$ and $C_2$ if it is the resolvent of the weakenings $C_1'$ and $C_2'$ of $C_1$ and $C_2$ respectively.
Ordinary resolution can be seen as a special case of weakened resolution where the weakenings of the premise clauses are $C_1' = C_1 \lor \bot$ and $C_2' = C_2 \lor \bot$. Weakened 
resolution proofs are sound and complete: they exist if and only if $\varphi$ 
is unsatisfiable \cite{LAH+22}. We write $C_1 \odot_{x} C_2$ to denote a weakened 
resolvent of $C_1$ and $C_2$ on variable $x$, treating the operator as left 
associative for chains: $C_{1} \odot_{x_{1}} C_{2} \odot_{x_{2}} C_{3} = 
(C_{1} \odot_{x_{1}} C_{2}) \odot_{x_{2}} C_{3}$.

\subsection{Cryptography}\label{sec:prelims-crypt}
We recall the necessary cryptographic notions, following \cite{ALP26, BG92, GoldreichFoC06, GMR85, HL10, LAH+22,thaler2022proofs, WYKW21}, and provide the formal definitions for the convenience of the reader in Appendix \ref{sec:appendix:formal-defs} following \cite{BG92, GoldreichFoC06, HL10}.

\paragraph{Interactive proofs.} An interactive proof system for a language $L$ is a protocol between two parties: a computationally unbounded prover $P$ and a probabilistic polynomial-time (PPT) verifier $V$, who exchange messages over multiple rounds. The system satisfies two properties: \emph{completeness}: for every $x \in L$, a verifier $V$ accepts with high probability; and \emph{soundness}: for every $x \notin L$ and every (possibly
cheating) prover, $V$ accepts with low probability. We refer the reader to Appendix \ref{sec:appendix:formal-def:IP}
for the formal definition.

\paragraph{Proofs of knowledge.} Conceptually, a \textit{proof of knowledge} is a protocol between two probabilistic polynomial time (PPT) parties, a prover $\mathcal{P}$ and a verifier $V$, where $\mathcal{P}$ convinces $V$ that for a given common input $\psi$ and relation $R \subseteq \{0,1\}^* \times \{0,1\}^*$, $P$ \textit{knows} a witness $w$ such that $(\psi, w) \in R$. A formal definition is included in Appendix \ref{sec:appendix:formal-def:PoK}.

\paragraph{Zero knowledge proofs.} A zero-knowledge proof system is an interactive proof with the additional property that the verifier learns nothing beyond the validity of the statement being proved. More formally, an interactive protocol between $P$ and $V$ is \emph{perfect zero-knowledge with respect to auxiliary input} if for every PPT verifier $V^*$, there exists a PPT machine $M^*$ (called the simulator) such that $\{M^*(x,z)\}_{x \in L,\, z \in \{0,1\}^*}$ (i.e., the output of machine $M^*$ on common input $x$ and auxiliary input $z$) is \emph{identically distributed} to $\{\langle P(y), V^*(z) \rangle(x)\}_{x \in L,\, z \in \{0,1\}^*}$ (i.e., the output of $V^*$ after it interacts with $P$ on common input $x$ and auxiliary inputs $y$ and $z$), where $y$ is the prover's private input.\footnote{As is standard, the simulator is permitted to output a special failure symbol $\bot$ with probability at most $\tfrac{1}{2}$, and the two distributions are required to agree conditioned on $M^*(x,z) \neq \bot$.} Intuitively, everything the verifier could compute from the interaction can be computed from the public input alone. The formal definition appears in Appendix \ref{sec:appendix:formal-def:ZKP}.

\paragraph{Commit-and-prove zero knowledge proofs.} Zero knowledge proofs can be constructed in multiple ways. \textsc{ZkUnsat} is constructed in the \textit{commit-and-prove} paradigm. A commit-and-prove zero knowledge proof proceeds in two phases. The first is a \textit{commitment phase} where given a field element $m$ in some finite field, the prover commits to it via a cryptographic commitment, denoted as $[m]$, while maintaining two key properties: \textit{(1) Hiding:} $[m]$ does not reveal $m$, and \textit{(2) Binding:} Except with negligible probability, $[m]$ cannot be the commitment of another message $m'$. Therefore, once a value $m$ is committed, it cannot be altered to be the commitment of another value. The second phase is the \textit{proving phase} where the prover proves a relation about the committed values.

\paragraph{Polynomial commitments.} Fix a large finite field $\F$ and public degree bound $d$, and let $\F_d[X]$ denote the polynomials over $\F$ of degree at most $d \in \mathbb{Z}^+$. A polynomial $p(X) \in \F_d[X]$ can be committed by committing to the $d+1$ coefficients of the polynomial. \textsc{ZkUnsat} instantiates commitments using a scheme based on Vector Oblivious Linear Evaluation (VOLE) \cite{YSWW21} that enables proving relations between these polynomials via their commitments. We denote this commitment of $p(X)$ by $[p(X)]$. In particular, given $p(X), q(X), r(X) \in \F_d[X]$, one can perform an equality check to prove the polynomial underlying $[p(X)]$ is equal to the polynomial underlying $[q(X)]$, and a factor check to prove $[r(X)]$ is a commitment of the polynomial $p(X) \cdot q(X)$. Given that resolution steps can be performed by treating clauses as sets (for example, $C=(l_1, \dots, l_n)$ can be treated as the set $\{l_1, \dots, l_n\}$), encoding clauses as polynomials, where a root of the polynomial corresponds to a literal in the clause, enables reasoning about clauses in a hiding and binding manner. Hereafter, we refer to the polynomial commitment scheme instantiated in \textsc{ZkUnsat} as the VOLE-based PCS.

\section{Related Work}\label{sec:background}

To provide a broader context for our work, we first briefly discuss the evolution of certification in the SAT community and then focus on recent efforts on zero-knowledge certification. 

\subsection{UNSAT Proof Formats}
The development of UNSAT proof formats has been a sustained effort in 
the SAT community~\cite{GN03,G08,WHH14,CHH+17,BCH22}, driven by two
complementary goals: minimizing proof-generation overhead on solvers, 
and enabling checkers to verify certificates efficiently.
Resolution proofs~\cite{ZM03} were the first independently checkable
certificates for SAT solvers. RUP (Reverse Unit Propagation)~\cite{GN03,G08} simplified proof generation
by allowing solvers to emit derived clauses without explicit resolution
chains, delegating justification reconstruction to the checker via unit
propagation.
The RAT (Resolution Asymmetric Tautology) property~\cite{JHB12} extended
this to clauses not logically implied by the formula, and the \texttt{RAT}
format~\cite{HHW13} enabled solvers to emit and justify such clauses.
DRAT~\cite{WHH14} augmented RAT with clause deletion, allowing solvers 
to signal when learned clauses are no longer needed; the accompanying
DRAT-trim tool reduced proof size by eliminating clauses irrelevant to 
the final refutation.
GRIT~\cite{CMS17} introduced efficient certified resolution proof checking by providing clause indices as hints, and allowing clause deletion, enabling verification in time proportional to the hint annotations rather than requiring expensive unit clause searches. Finally, LRAT \cite{CHH+17} built on GRIT's hint-based approach, extending it to the RAT proof system and enabling an efficient formally verified proof checker, implemented in ACL2.  These combined efforts have led to faster LRAT checking than solving \cite{PFB23} with CaDiCaL \cite{BFF+24}, striking a balance between efficient proof generation by the SAT solver and efficient verification by the proof checker. The striking improvements in proof checking are a core motivation of our work as we seek to achieve similar improvements for proof checking with zero-knowledge. 

\subsection{Zero-Knowledge UNSAT Certification}
Building on these advances in efficient plaintext checking, recent work seeks to extend comparable certification guarantees to privacy-sensitive settings via zero-knowledge verification of UNSAT. \textsc{ZkUnsat} \cite{LAH+22} established the first practically viable protocol of this kind, proving the validity of a weakened resolution proof in zero knowledge. More recently, ZK-ProVer \cite{KWLL26} proposed a non-interactive alternative based on Scalable Transparent ARguments of Knowledge (zk-STARKs) \cite{BBHR19}. While ZK-ProVer substantially reduces verifier time and communication, it does so at a steep cost in prover memory: across their benchmarks, \textsc{ZkUnsat} exhibits markedly flatter and lower memory utilization, in some configurations as much as $10\times$ lower than the ZK-ProVer construction. Since prover memory is the dominant scalability bottleneck we target, \textsc{ZkUnsat} remains the natural baseline for zero-knowledge UNSAT certification.

\subsubsection{ZKUNSAT}

Let $\varphi$ be a CNF formula and $\pi$ a chained weakened resolution
proof of its unsatisfiability.
\textsc{ZkUnsat}~\cite{LAH+22} is an interactive
proof protocol in which a prover $\mathcal{P}$ holding $\pi$ convinces
a verifier $\mathcal{V}$ that $\varphi$ is unsatisfiable.
When $\pi$ is fully unfolded prior to execution, the protocol is
zero-knowledge: $\mathcal{V}$ learns nothing beyond the existence of a
weakened resolution proof, the length of the unfolded proof, and the size of the largest clause appearing in any step of the proof.

\textsc{ZkUnsat} encodes clauses as polynomials over a finite field
$\mathbb{F}$.
Luo et al.\ construct an encoding $\phi$ of clauses into the set of polynomials
$\mathbb{F}_{d}[X]$ where $d \in \mathbb{Z}_{+}$ is the size of the largest clause appearing in the proof. The encoding is such that verifying $C_{\text{wres}} = C \odot_{x} C'$
reduces to checking algebraic identities involving $\phi(C)$, $\phi(C')$,
$\phi(x)$, and two witness polynomials $w, w'$.
We write $[C] := [\phi(C)]$ to denote the commitment of the polynomial encoding of the clause $C$.
We refer the reader to~\cite{LAH+22} for a complete treatment of $\phi$.

\textsc{ZkUnsat} is analyzed in the Universal Composability (UC) framework \cite{Canetti2000, Canetti20, CLOS02, HL10}, where each cryptographic primitive is modeled as an \emph{ideal functionality}: a trusted third party that captures the primitive's intended security guarantees. The protocol's security is established by showing that a real execution is indistinguishable from one in which these ideal functionalities are present. \textsc{ZkUnsat} relies on three such functionalities, whose formal specifications appear in Appendix \ref{sec:appendix:functionalities}:
\begin{itemize}
    \item \Func[ZK] handles zero-knowledge commitments and polynomial checks. It allows the prover to commit secret field elements (via $\sf{Witness}$), both parties to authenticate shared values (via $\sf{Instance}$), and both parties to verify circuit relations and product-of-polynomial equality checks over committed values.
    \item \Func[Clause] provides clause-level operations in zero knowledge. It supports committing a clause with a width bound (via $\sf{Input}$), checking equality of two committed clauses (via $\sf{Equal}$), verifying that a committed clause is a weakened resolvent of two others (via $\sf{X-Res}$), and asserting that a committed clause is the empty clause (via $\sf{IsFalse}$).
    \item \Func[FlexZKArray] implements a zero-knowledge random-access array, hereafter referred to as the \emph{ZkRam}\cite{FKL+21}. It is initialized with a sequence of committed values and supports private-index reads: the prover retrieves an entry at a secret index while the verifier learns only an upper bound on that index. A deferred $\sf{check}$ operation verifies the consistency of all preceding accesses, detecting any attempt by the prover to read a value that differs from the one originally stored.
\end{itemize}

The protocol is built on the clause encoding $\phi$, the VOLE-based PCS, and the ZkRam $\mathcal{R}$, and follows the commit-and-prove paradigm (Section \ref{sec:prelims-crypt}).
 
\paragraph{Commit phase.}
The $|\varphi|$ clauses of the input formula are committed through $\Func[Clause]$ and the $d+1$ coefficients of each clause are confirmed to represent the clauses in $\varphi$ using the $\sf{Instance}$ operation from $\Func[ZK]$ to authenticate the coefficients. The $|\pi|$ derived clauses in the fully unfolded proof, known only to $\mathcal{P}$, are committed via $\Func[Clause]$'s $\sf Input$ operation as degree $d$ polynomials. The \textsc{ZkRam} $R$ then initializes all $|\varphi| + |\pi|$ commitments through $\Func[FlexZKArray]$.
 
\paragraph{Proving phase.} The protocol iterates over the $|\pi|$ steps of the refutation. In each iteration $i$, the prover privately retrieves two premise clauses $C_{k_i}$ and $C_{l_i}$ from $\mathcal{R}$ via $\Func[FlexZKArray]$'s $\sf{Read}$ operation---so that the verifier learns only an upper bound on the accessed indices---and then both parties invoke $\Func[Clause]$'s $\sf{X\text{-}Res}$ operation to verify that the $i$-th derived clause $C_i$ is a valid weakened resolvent of $C_{k_i}$ and $C_{l_i}$. After all $|\pi|$ iterations, both parties call $\sf{IsFalse}$ on the final clause to confirm it equals $\bot$, and invoke $\Func[FlexZKArray]$'s deferred $\sf{check}$ to verify the consistency of all preceding array accesses. \\
 
\noindent Note that $|\pi|$ is revealed as it is the number of iterations in the
proving phase, and that the size $d$ of the largest clause is revealed
as part of the specification of the PCS.

\paragraph{Security.} The setting in which one or more subprotocols are replaced by their ideal functionalities $(\Func[1], \dots, \Func[m])$ is called a $(\Func[1], \dots, \Func[m])$-\textit{hybrid model}. The authors of \cite{LAH+22} prove that \textsc{ZkUnsat} is a zero knowledge proof of knowledge (in the $(\Func[ZK], \Func[Clause], \Func[FlexZKArray])$-hybrid model) in the full version \cite{LAH+22-ext}.

\paragraph{Scaling.} Recent work \cite{KLM+26} on proving quantified Boolean formulas in zero knowledge scales \textsc{ZkUnsat} by introducing a clause-partitioning scheme that buckets clauses of similar width and commits to them as lower-degree polynomials. While this optimization improves the runtime of \textsc{ZkUnsat} by approximately $50\%$, it accepts an explicit leakage-efficiency tradeoff: bucketing reveals structural information about the clause-width distribution. We do not adopt this optimization, as our goal is to reduce prover memory without introducing the structural leakage it incurs.

\section{Identifying and Addressing the Unfolding Bottleneck}\label{sec:observations}

\textsc{ZkUnsat} expects a weakened resolution proof in which each proof line encodes a single resolution step. To meet this requirement, proofs produced by SAT solvers must be fully unfolded: every derived clause is produced from exactly two parent clauses, substantially inflating proof size. In this section, we show that this inflation is the dominant scalability bottleneck in \textsc{ZkUnsat}, and present a preprocessing algorithm that normalizes each resolution chain to a fixed public length k -- capturing the memory benefits of chaining without introducing additional leakage.

\subsection{The Memory Bottleneck from Unfolding}\label{sec:mem-bottleneck-from-unfold}

Unfolding a chained resolution step of length $m$ introduces $m-1$ intermediate clauses. Each of these requires its own index and \textsc{ZkRam} commitment, so the prover's memory footprint grows with the fully unfolded proof size. To quantify this cost, we ran \textsc{ZkUnsat} directly on chained resolution proofs, bypassing unfolding entirely. We refer to this variant as Var-Chain-\textsc{ZkUnsat}. On the 510 proofs we obtain from the SAT 2002 benchmark suite (Section~\ref{sec:evaluation}), \textsc{ZkUnsat} verified 220 of 510 instances, with 290 aborted due to memory exhaustion at the 32 GB virtual memory limit within $5{,}000$ seconds. Var-Chain-\textsc{ZkUnsat} verified 281, with only 103 memory aborts and 126 timeouts with the same resource limits. As Figure~\ref{fig:rt_vs_mem_ZKUNSAT_vs_var-chain_ZKUNSAT} shows, \textsc{ZkUnsat}'s runtime and peak memory are tightly correlated, consistent with memory being the primary bottleneck. Var-Chain-\textsc{ZkUnsat} exhibits no such pattern, and its peak memory stays below 15 GB on most instances.
Var-Chain-\textsc{ZkUnsat}, however, reveals the length of each resolution chain to the verifier. In the next section, we show how to recover the memory benefits of chaining 
without this leakage.

\begin{figure}
    \centering
    \input{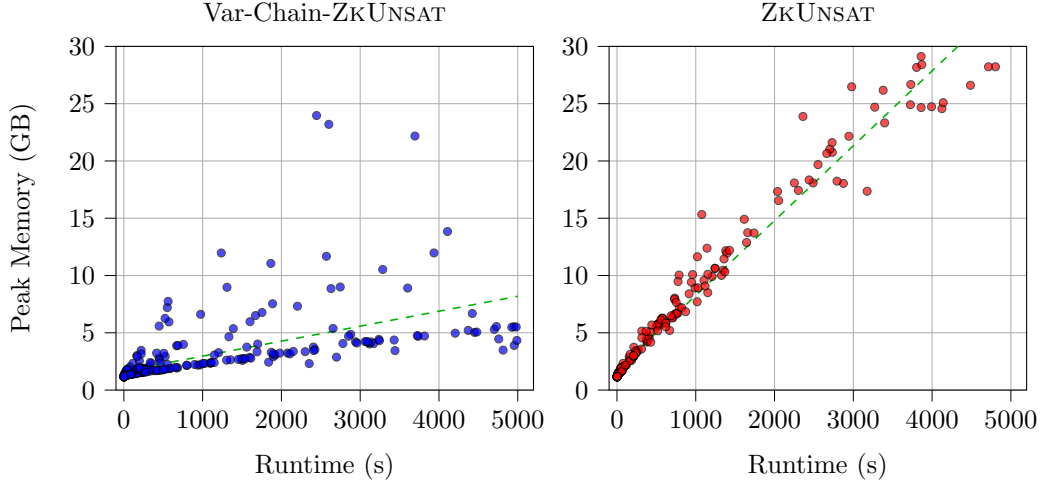}
    \caption{Runtime vs.\ prover peak memory usage on the SAT 2002 benchmarks for Var-Chain-\textsc{ZkUnsat} and \textsc{ZkUnsat}.}
    \label{fig:rt_vs_mem_ZKUNSAT_vs_var-chain_ZKUNSAT}
\end{figure}

\subsection{Fixed-Length Chain Normalization}\label{sec:chain-normalization}

We present a preprocessing algorithm that converts an input chained resolution proof -- which we refer to as the \emph{parent proof} -- into a chained weakened resolution proof in which every chain has length exactly $k$, exploiting the surprisingly useful fact that the weakened resolution of a clause $C$ with itself on any pivot variable $x$ admits $C$ as a weakened resolvent. The algorithm processes each proof line incrementally while maintaining an index mapping, \texttt{map}, that takes a clause index from the old proof and provides its corresponding index in the new proof. If the hints on a line are $I_1, \dots, I_m$, it invokes Normalize($k$, $h_1, \dots, h_m$) (Alg.~\ref{alg:normalize}) where $h_i = \texttt{map}[I_i]$.

For a fixed $k$, we call this combined pipeline---normalizing a parent proof into one where every chain has length $k$, then running \textsc{ZkUnsat}---$k$-Chain-\textsc{ZkUnsat} (a formal description of the protocol is included in Appendix \ref{sec:appendix-protocol}). Note that 1-Chain-\textsc{ZkUnsat} is precisely \textsc{ZkUnsat} on fully unfolded proofs. We also note that given a parent proof with chain lengths $L_1, \dots, L_n$, the Normalize algorithm produces a weakened resolution proof containing precisely $\sum_{i=1}^n k \cdot \lceil L_i/k \rceil$  weakened resolution steps. Therefore, the public parameters of $k$-Chain-\textsc{ZkUnsat} are: the input formula $\varphi$, $k$, the number of chains in the normalized proof $|\pi|=\sum_{i=1}^n\lceil L_i/k \rceil$, and the degree $d$ of the polynomials used to represent the clauses in the proof. 

\begin{algorithm}[t]
	\caption{Normalize($k, (h_1, \dots,h_m)$)}
	\label{alg:normalize}
	\begin{algorithmic}[1]
        \If{$m \le k+1$}
            \State $C \gets$ Resolve($h_1,\dots,h_m$)
            \For{$i \in \{1,\dots,k+1-m\}$}
            \State $h_{1,i} \gets h_1$
            \EndFor
            \State \label{padding}\Call{Print}{$C,(h_{1,1},\dots,h_{1,(k+1-m)},h_1,\dots,h_m)$}
        \Else \Comment{$m > k+1$}
            \State $C \gets $ Resolve($h_1, \dots ,h_{k+1}$)
            \State\Call{Print}{$C,(h_1,\dots,h_{k+1})$}
            \State \label{chop}Normalize($k$, (Index($C$), $h_{k+2}, \dots, h_m$))
        \EndIf
	\end{algorithmic}    
\end{algorithm}

\section{Security and Memory Analysis}\label{sec:Thm-and-privacy-preservation}

In this section, we state the central security guarantee of $k$-Chain-\textsc{ZkUnsat}, that it constitutes a zero-knowledge proof of knowledge of a normalized weakened chained resolution proof, with the accompanying proof given in Appendix \ref{sec:appendix:sec-proof}. We also identify two usage restrictions that must be observed to maintain this guarantee, namely, $k$ must be selected a priori, independent of the proof structure, and the protocol must be executed for only a single value of $k$ on any given parent proof. We then derive a peak memory estimator that allows practitioners to predict, before committing to a full execution, whether a given (instance, $k$) pair will fit within a memory budget.

\subsection{Security Analysis}
The high-level proof structure of \cite[Theorem 1]{LAH+22-ext}---establishing that the \textsc{ZkUnsat} protocol constitutes a zero-knowledge proof of knowledge of a refutation---generalizes to $k$-Chain-\textsc{ZkUnsat}. We state this formally below and include a proof in Appendix \ref{sec:appendix:sec-proof} for the convenience of the reader.
 
\begin{theorem}\label{thm:k-chain-zk}
$k$-Chain-\textsc{ZkUnsat} is, against static corruption, a zero-knowledge proof of knowledge of a length-$k$ chained weakened resolution proof, with common input $\psi = (\varphi,\; k,\; |\pi|,\; d)$, where $\varphi$ is the input formula, $k$ the chain bound, $|\pi|$ the number of chains in the normalized proof, and $d$ the degree of the polynomials encoding the clauses in \textsc{ZkUnsat}'s $(\Func[ZK], \Func[Clause], \Func[FlexZKArray])$-hybrid model.
\end{theorem}

We note two important usage restrictions regarding k-Chain-\textsc{ZkUnsat}. First, for any given parent proof there exists a $k$ such that the peak memory is minimal for that parent proof, and the runtime is minimal for that peak memory usage. This is primarily because the prover's peak memory scales with the total number of clauses committed within the \textsc{ZkRam}, which is minimal when $k \ge \max_i\{L_i\}$, and given this constraint on $k$, the $k$ that minimizes the number of operations to be proven/verified is $k = \max_i\{L_i\}$. Such values of $k$ that depend on the parent proof should not be used in practice, as their values may reveal information about the underlying $L_i$. We therefore recommend selecting $k$ \emph{a priori}, without knowledge of the proof structure, to prevent such leakage.

Second, one should not execute k-Chain-\textsc{ZkUnsat} for multiple values of $k$ on the same parent proof if the k-Chain proofs were produced by the Normalize algorithm, as doing so may leak information about the parent proof's chain length distribution. For instance, if k-Chain-\textsc{ZkUnsat} is executed for all $k \in \{1, \dots, p\}$, where $p$ is the total number of resolution steps as revealed by 1-Chain-\textsc{ZkUnsat}, this would reveal the number of chains of length $l$ for every $1 \le l \le p$, which is formalized in the following proposition (a proof is included in Appendix \ref{sec:appendix:sec-proof}).

\begin{proposition} \label{prop:main-paper}
    Given a parent proof, if the k-Chain-\textsc{ZkUnsat} protocol is executed for all $k \in \{1, \dots, p\}$ where $p$ is the total number of resolution steps in the parent proof, and the k-Chain weakened resolution proofs are derived using the Normalize algorithm, then the collection of all executions reveals the number of chains of length $l$ for every $1 \le l \le p$.
\end{proposition}

\subsection{Deriving a Peak Memory Estimator}\label{sec:estimator-details}

While \textsc{ZkUnsat} provides reliable runtime estimates, no comparable prediction exists for its peak memory consumption. We address this gap by analyzing the implementation and deriving a lightweight heuristic $E_m$, that allows a user to predict, before committing to a full protocol execution, whether a given normalized proof will fit within a memory budget. The estimator depends on the proof parameters $|\varphi|$ (the number of input clauses), $|\pi|$ (the number of chains in the normalized proof), $k$ (the chain length), and the largest clause degree $d$, together with two \textsc{ZkUnsat} configuration constants: the \textsc{ZkRam} index bit-width $I$ (so the maximum table size is $2^{I-1}$) and the consistency-check batch size $c_{\varphi}$, detailed below.

\noindent\textbf{Batched Consistency Checks.} The \textsc{ZkRam} maintains an \emph{access record} that is initialized with the $|\varphi|+|\pi|$ clauses in the proof. Every time a clause is subsequently retrieved, the clause and its index are appended to this log. To amortize the cost of verifying the log's consistency, the implementation does not check after every access. Instead, it executes a \texttt{check()} routine that sorts the accumulated record entries, verifies that every accessed index maps to the correct clause, and then truncates the record back to its initial size $|\varphi| +|\pi|$. The routine is invoked automatically once every $c_{\varphi}$ accesses, where $c_{\varphi} = 2(|\varphi|+|\pi|)$ by default.

\noindent\textbf{Counting Accesses and Checks.} In $k$-Chain-\textsc{ZkUnsat}, each normalized chain triggers $(k+2)$ \textsc{ZkRam} accesses: one for each of the $k+1$ clauses in the chain, plus one final access to assert that the resulting clause matches the entry stored in the \textsc{ZkRam}. Over all $|\pi|$ chains, the total number of accesses is $(k+2)\cdot|\pi|$, and the number of access-triggered \texttt{check()} invocations is
$N_{\mathsf{chk}} =\left\lfloor \frac{(k+2) \cdot |\pi|}{c_{\varphi}} \right\rfloor$.

\begin{table}[t]
    \centering
    \renewcommand{\arraystretch}{1.2}
    \begin{tabular}{@{} l l @{}}
        \toprule
        \textbf{Data Structure} & \textbf{Size (in bits)} \\
        \midrule
        \texttt{sorted\_clear\_access}   & $(n)(64 + 64d)$ \\
        \texttt{clear\_access\_record}   & $(n)(64 + 64d)$ \\
        \addlinespace
        \texttt{access\_record} (index)  & $(n)(128I)$ \\
        \texttt{access\_record} (clause) & $(n)(d(128\cdot 2+64))$\\
        \texttt{sorted\_index}           & $(n)(128I)$ \\
        \texttt{sorted\_clause}          & $(n)(d(128\cdot 2+64))$\\
        \addlinespace
        \texttt{HRecord}                 & $(n)(128)$ \\
        \texttt{HRecord\_mac}            & $(n)(128)$ \\
        \texttt{sorted\_HRecord}         & $(n)(128)$ \\
        \texttt{sorted\_HRecord\_mac}    & $(n)(128)$ \\
        \addlinespace
        \texttt{sorted\_hash\_value}     & $(n)(128^2)$ \\
        \bottomrule
    \end{tabular}
    \caption{Largest data structures allocated during \texttt{check()}, where $n$ denotes the number of entries in the access record at invocation time.}
    \label{tab:zkram_data_structures}
\end{table}

Table~\ref{tab:zkram_data_structures} details the size of the primary data structures allocated during this operation. By summing the bit-sizes from Table~\ref{tab:zkram_data_structures} and converting the total to 128-bit blocks (a standard metric for emp-tool \cite{emp-toolkit} structures), we arrive at our peak memory estimator. We must account for two scenarios: if $N_{\mathsf{chk}} \ge 1$, peak memory occurs during an intermediate \texttt{check()} with exactly $|\varphi|+|\pi|+c_{\varphi}$ elements, and if $N_{\mathsf{chk}} = 0$, the number of accesses never reaches the consistency-check batch size $c_{\varphi}$, and the sole \texttt{check()} occurs at the end of the protocol with $|\varphi| + |\pi|+ (k+2)\cdot|\pi|$ elements. Therefore, the peak memory usage in 128-bit blocks is:

$$\boxed{E_m = \begin{cases}
(|\varphi|+|\pi|+c_{\varphi})\cdot(2I+133+6d) & \text{if } N_{\mathsf{chk}} \ge 1 \\
(|\varphi|+|\pi|+|\pi|(k+2))\cdot(2I+133+6d) & \text{otherwise}
\end{cases}}$$

We emphasize that $E_m$ is an implementation-level heuristic tied to the current \textsc{ZkUnsat} codebase, not a fundamental bound on the memory requirements of zero-knowledge UNSAT certification. Changes to the \textsc{ZkRam}'s batching strategy or internal data structures would require re-deriving the estimator. Nevertheless, $E_m$ serves a  concrete and practical role: it allows us to predict, before committing to a full protocol execution, whether a given (instance, $k$) pair will fit within a memory budget.

\section{Experimental Evaluation}\label{sec:evaluation}
We evaluate $k$-Chain-\textsc{ZkUnsat} empirically, addressing four questions:
(1)~How many additional instances can $k$-Chain-\textsc{ZkUnsat} certify compared to the baseline \textsc{ZkUnsat} under realistic resource constraints?
(2)~How do runtime and peak memory vary with $k$?
(3)~How should one select $k$ in practice?
(4)~How well does the peak memory estimator $E_m$ predict the actual peak memory of $k$-Chain-\textsc{ZkUnsat} across different values of $k$?

Section~\ref{sec:exp-setup} describes our experimental setup, Section~\ref{sec:results} presents performance results for $k$-Chain-\textsc{ZkUnsat}, and Section~\ref{sec:peak-mem-results} evaluates the peak memory estimator.

\subsection{Setup}\label{sec:exp-setup}
For all experiments, we use a cluster equipped with two 96-core AMD EPYC 9655 processors running at a 2.6 GHz base frequency, with 810 GB of memory per node. To maximize throughput, we deploy 10 \textsc{ZkUnsat} prover–verifier pairs concurrently per node, limiting the virtual memory of each prover and verifier process individually to 32GB using \texttt{ulimit}. We also co-locate each prover-verifier pair on the same node and route their protocol communication via localhost. 

To generate proofs and perform evaluations, we pick the SAT Competition 2002 benchmark suite. We select this benchmark suite over the benchmarks in \cite{LAH+22} because the latter demonstrates the scalability of \textsc{ZkUnsat} on only 58 instances---a sample too small to support statistically meaningful conclusions for our research questions. Moreover, the SAT 2002 benchmarks represent an achievable frontier for \textsc{ZkUnsat} as the original \textsc{ZkUnsat} already exhausts memory on the majority of these instances. Modern competition benchmarks, which tend to produce substantially larger proofs, would yield even fewer certifiable instances, limiting the scope for meaningful comparison. We obtain 1,964 CNF formulas by combining both the generated and the submitted instance sets. We solve these instances with CaDiCaL v2.1.3 \cite{BFF+24} under a 3,600s timeout, with UNSAT certification enabled via the flags $\texttt{{-}{-}lrat}$ and $\texttt{{-}{-}binary=false}$. CaDiCaL returns UNSAT certificates for 554 instances. 

For these 554 instances, we post-process the produced LRAT proofs using lrat-trim. With a virtual memory limit of 8\,GB, lrat-trim successfully produced 548 trimmed proofs. These trimmed proofs are the input to our preprocessing pipeline, which consists of two Python scripts.\footnote{Available at \url{https://github.com/meelgroup/chain-normalize}} The first script is a modified version of the \texttt{ExtendProof.py} preprocessor shipped with \textsc{ZkUnsat}, which converts a formula and its trimmed LRUP-style proof into a chained resolution proof in the format expected by \textsc{ZkUnsat}: it reverses the hint order of every proof line and recovers the pivot of each resolution step. This conversion succeeded on 510 instances. The resulting chained resolution proofs serve as the parent proofs for our evaluation. Our second script is an implementation of the Normalize algorithm (Algorithm~\ref{alg:normalize}), which given a parent proof and a chain bound $k$, normalizes every chain to $k$-Chain(s) and emits the resulting $k$-Chain weakened resolution proof, which is then passed to the \textsc{ZkUnsat} prover.

We consider five preprocessing configurations with chain bound $k \in \{1, 3, 5, 16, 32\}$. Note that the configuration $k = 1$ corresponds precisely to the baseline \textsc{ZkUnsat}. We select $k \in \{3, 5, 16, 32\}$ since excessively large values of $k$ introduce a significant number of padded hints, increasing runtime without further reducing memory. We set the timeout to be $25{,}000$ seconds. 

Since \textsc{ZkUnsat} fixes the bit-width used to represent clause indices via the macro $\texttt{index\_sz}$ (set to 20 by default), which effectively caps the maximum clause index at $2^{19}=524{,}288$, \textsc{ZkUnsat}'s default configuration cannot accommodate several of our larger proofs. Since smaller values of $\texttt{index\_sz}$ yield better runtime performance, we increased this parameter to $\texttt{index\_sz}=22$ after consulting a developer of \textsc{ZkUnsat}, raising the limit to $2^{21}=2{,}097{,}152$ clause indices.\footnote{Even with this increase, an error arose in baseline \textsc{ZkUnsat} since one proof contained $4{,}187{,}924$ indexed clauses; increasing $\texttt{index\_sz}$ to 24 for this instance caused the prover to abort due to memory exhaustion.} Apart from this change, the prover and verifier of the \textsc{ZkUnsat} implementation are unmodified.

\begin{table}[t]
    \centering
    \begin{tabular}{ccccc}
      \toprule
      \textbf{Variant} & \textbf{Verified} & \textbf{Aborts (OOM)} & \textbf{Timeouts} \\
      \midrule
      \textsc{ZkUnsat}  & 220 & 290 & 0 \\
      k=3 & 276 & 234 & 0 \\
      k=5 & 317 & 193 & 0 \\
      {\bfseries k=16} & {\bfseries 357} & {\bfseries 146} & {\bfseries 7} \\
      k=32 & 339 & 121 & 50 \\
      \bottomrule
    \end{tabular}
    \caption{Performance of different variants of \textsc{ZkUnsat} on the SAT 2002 competition benchmarks (timeout: 25,000\,s; virtual memory limit per prover and verifier process: 32\,GB; \texttt{index\_sz}=22).}
    \label{tab:summary-25000s}
\end{table}

\subsection{Results}\label{sec:results}
Table \ref{tab:summary-25000s} summarizes the performance of $k$-chain-\textsc{ZkUnsat} for varying values of $k$. For each $k \in \{1, 3, 5, 16, 32\}$, the 510 test instances are split across three columns: Verified, Aborts (OOM) and Timeouts. The verified column notes how many instances were verified within the time and memory constraints. The Aborts (OOM) column notes the number of instances on which the protocol exceeded the memory limit, and the Timeouts column notes the number of instances on which the protocol ran out of time. Note that the first row shows the performance of the original $\textsc{ZkUnsat}$ protocol, which corresponds to $k = 1$. We take this to be our baseline.

\noindent\textbf{Overall Performance.} Baseline \textsc{ZkUnsat} successfully certifies 220 of the 510 instances, with 290 aborted due to memory exhaustion and none timing out. In contrast, $16$-Chain-\textsc{ZkUnsat} certifies 357 instances, an improvement of approximately 62\%, while reducing the number of memory-related aborts from 290 to 146. This confirms that the chain normalization technique in Section \ref{sec:chain-normalization} effectively addresses the memory bottleneck of \textsc{ZkUnsat}.

\begin{figure}
    \centering
    \input{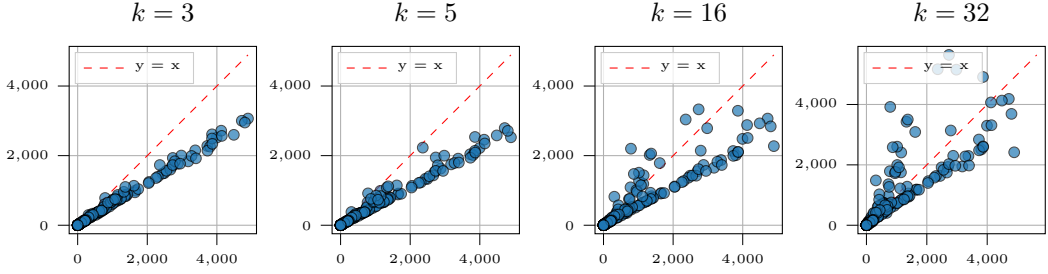}
    \caption{Runtime (s) of $k$-Chain-\textsc{ZkUnsat} (y-axis) against runtime (s) of
    baseline \textsc{ZkUnsat} (x-axis), restricted to instances verified by baseline
    \textsc{ZkUnsat}. Each panel corresponds to a different value of $k$. Points below
    the diagonal represent instances with faster verification under $k$-Chain-\textsc{ZkUnsat}.}
    \label{fig:scatter_plot_rt}
\end{figure}

\noindent\textbf{Effect of Varying k.} Increasing $k$ from 1 to 32 monotonically reduces memory-related aborts: $k = 3$ reduces aborts to 234, $k = 5$ to 193, $k = 16$ to 146, and $k = 32$ to 121. However, larger values of $k$ introduce additional padding operations per chain, which increases runtime. This trade-off is visible in the timeout column of Table~\ref{tab:summary-25000s}: while $k \in \{3, 5\}$ incur no timeouts and $k = 16$ incurs only 7, $k = 32$ times out on 50 instances under the $25{,}000$\,s limit. As a result, 32-Chain-\textsc{ZkUnsat} verifies fewer instances (339) than 16-Chain-\textsc{ZkUnsat} (357), despite having fewer memory aborts. This demonstrates that $k = 16$ strikes the best balance between memory savings and runtime overhead on this benchmark suite.

\noindent\textbf{Runtime Comparison.} Figure \ref{fig:scatter_plot_rt} shows scatter plots comparing the runtime of $k$-Chain-\textsc{ZkUnsat} (y-axis) against baseline \textsc{ZkUnsat} (x-axis) for each $k \in \{3,5, 16,32\}$, restricted to instances verified by baseline \textsc{ZkUnsat}. Note that every such instance is also verified by $k$-Chain-\textsc{ZkUnsat} for all values of $k$ considered. Points below the diagonal correspond to instances where $k$-Chain-\textsc{ZkUnsat} is faster. 
\begin{itemize}
    \item For $k = 3$ and $k = 5$, the vast majority of points lie below the diagonal, indicating consistent speedups across the benchmark suite.
    \item For $k = 16$, most points remain below the diagonal, though a small number of instances exhibit comparable or slightly higher runtimes.
    \item For $k = 32$, several points shift above the diagonal, reflecting the increased cost of padding short chains to length 32.
\end{itemize}

\begin{figure}[tb]
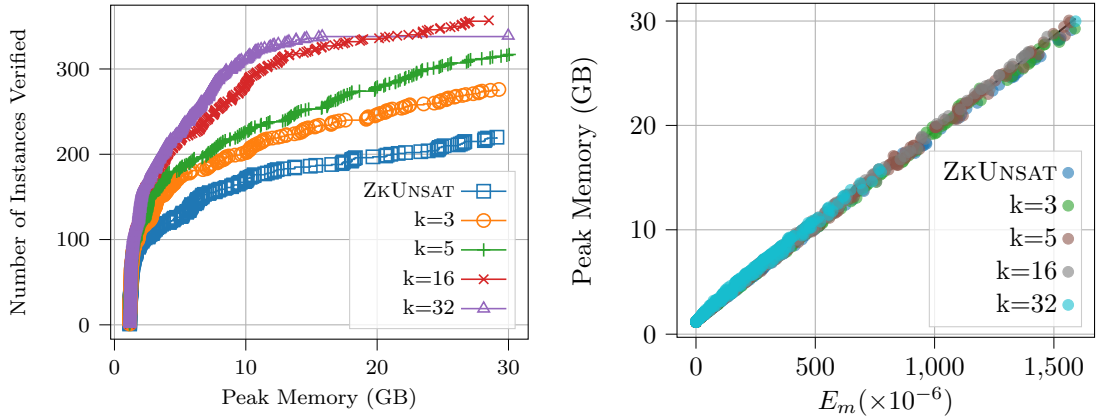

    \centering
    \begin{subfigure}[t]{0.49\textwidth}
        \centering
        \input{plots/cdf_plot_mem_vs_num_instances_25000s}
        \caption{Peak memory comparison of \textsc{ZkUnsat} against $k$-Chain-\textsc{ZkUnsat} for $k \in \{3,5,16,32\}$ on the SAT 2002 Competition benchmarks.}
        \label{fig:cdf_plot_mem_vs_num_instances_25000s}
    \end{subfigure}
    \hfill
    \begin{subfigure}[t]{0.49\textwidth}
        \centering
        \vbox to 5.5cm{\input{plots/memory_estimator}\vfil}
        \caption{Correlation between the memory estimator $E_m$ (with $c_\varphi = 2(|\varphi|+|\pi|)$) and the prover's peak memory usage. The line of best fit achieves $R^2 = 0.9984$.}
        \label{fig:memory_est}
    \end{subfigure}
    \caption{Memory usage and estimation for \textsc{ZkUnsat} and $k$-Chain-\textsc{ZkUnsat} for $k \in \{3, 5, 16,32\}$ on the SAT 2002 Competition benchmarks.}
    \label{fig:memory_combined}
\end{figure}

\noindent\textbf{Memory Comparison.} Figure \ref{fig:cdf_plot_mem_vs_num_instances_25000s} shows the number of instances verified as a function of peak prover memory usage for each configuration. A point $(x, y)$ indicates that $y$ instances were verified with a peak memory usage of at most $x$ GB. We see that while the curve for baseline \textsc{ZkUnsat} plateaus at 220 instances, 16-Chain-\textsc{ZkUnsat} verifies 220 instances---the same count as baseline ZkUnsat's total---using less than 25\% of the peak memory that baseline \textsc{ZkUnsat} required.

\noindent\textbf{Shifting the Bottleneck.} A key qualitative observation is that chain normalization shifts the performance bottleneck from memory to runtime. For baseline \textsc{ZkUnsat}, all 290 failures are caused by memory exhaustion. For 16-Chain-\textsc{ZkUnsat}, only 146 of 153 failures are memory-related, with 7 due to timeouts, and finally for 32-Chain-\textsc{ZkUnsat}, timeouts account for 50 of 171 failures, while memory aborts drop to 121.

\subsubsection{Selecting k}
The chain length parameter $k$ must be fixed a priori to preserve the zero-knowledge guarantees established in Theorem \ref{thm:k-chain-zk}, so we conclude with practical guidance for its selection. Our results reveal a clear trade-off governed by two competing effects: increasing $k$ reduces the prover's peak memory footprint by avoiding intermediate clause materialization, but simultaneously increases runtime due to the padding of chains whose length is not a multiple of $k$. This tension is quantified in Table~\ref{tab:summary-25000s}, where moving from $k=1$ to $k=16$ monotonically increases the number of verified instances (from 220 to 357) as memory aborts drop. However, beyond this point, the padding overhead dominates -- $k=32$ suffers 50 timeouts despite having the fewest memory aborts (121), ultimately certifying fewer instances than $k=16$. Figure~\ref{fig:scatter_plot_rt} corroborates this pattern at the per-instance level, where for $k=3$ and $k=5$, nearly all points fall below the diagonal, indicating consistent speedups, whereas for $k=16$ and $k=32$ a growing number of instances shift above it.

More broadly, the best choice of $k$ depends on the resource profile of the deployment environment. In time-restricted settings where wall-clock budget is the binding constraint, our results suggest $k \le 16$: values in this range yield substantial memory savings while introducing few or no timeouts, whereas in memory-restricted settings, where the prover must operate under a tight memory ceiling, larger values of $k$ are preferable, since each increase in $k$ reduces the number of intermediate clauses committed to the \textsc{ZkRam}. This is corroborated in  Figure~\ref{fig:cdf_plot_mem_vs_num_instances_25000s}: 16-Chain-\textsc{ZkUnsat} matches the baseline's total of 220 verified instances using less than 25\% of the maximum peak memory used by baseline \textsc{ZkUnsat}, leaving considerable headroom for instances that the baseline could not attempt.

\subsection{Peak Memory Estimation}\label{sec:peak-mem-results}
To evaluate the accuracy of the peak memory estimator $E_m$ derived in Section \ref{sec:estimator-details}, we compare the predicted values of $E_m$ against the observed peak memory usage of the prover across all configurations $k \in \{3,5,16,32\}$ and the baseline \textsc{ZkUnsat}.

Figure \ref{fig:memory_est} demonstrates the strong correlation between our estimator $E_m$ (with $c_{\varphi} = 2(|\varphi|+|\pi|)$) and the protocol's peak memory usage across all configurations. Each point corresponds to a single (instance, $k$) pair, and the five configurations are distinguished by marker style. The line of best fit achieves an $R^2 = 0.9984$, showing that the estimator is a near-perfect linear predictor of peak memory across the full range of configurations and instance sizes. 

Notably, the line of best fit holds uniformly across all our configurations. In particular, the estimator remains accurate both for the baseline \textsc{ZkUnsat}, where a significant number of instances used more than 20 GB, and for 32-Chain-\textsc{ZkUnsat}, where almost all instances used less than 20 GB, suggesting that $E_m$ can serve as a lightweight feasibility filter, allowing practitioners to skip doomed executions and allocate computational resources only to instances likely to complete within a given memory budget.

\section{Conclusion}\label{sec:conclusion}
We presented $k$-Chain-\textsc{ZkUnsat}, an approach to scaling zero-knowledge
certification of Boolean formula unsatisfiability.
Our analysis showed that proof unfolding---expanding chained resolution steps
into binary resolution sequences---is the main source of memory overhead in
\textsc{ZkUnsat}, a finding we confirmed via the Var-Chain variant.
To fix this without losing zero-knowledge guarantees, we introduced a
normalization algorithm that rewrites any chained resolution proof so that
every chain has length exactly $k$, and proved that the resulting protocol
reveals nothing beyond the number of proof lines and the public parameter $k$.
On the SAT 2002 benchmarks, $k = 16$ certifies roughly 62\% more instances
than baseline \textsc{ZkUnsat}, and for an equivalent number of certified instances, 
keeps peak prover memory below 25\% of what the baseline requires---moving the bottleneck from memory to runtime.
We also developed a peak memory estimator $E_m$ that achieves $R^2 = 0.9984$,
giving practitioners a simple way to predict whether an instance will fit in
memory before running the full protocol.

\section{Acknowledgement} 
We thank Arijit Shaw, Nikolay Avramov, Hidde Lycklama, Alexander Viand, and Ning Luo for many useful discussions. This work was partially supported by the Natural Sciences and Engineering Research Council of Canada (NSERC) through a Discovery Grant [RGPIN-2025-06535, RGPIN-2024-05956]. Computations were performed on the Trillium supercomputer at the SciNet HPC Consortium. SciNet is funded by Innovation, Science and Economic Development Canada; the Digital Research Alliance of Canada; the Ontario Research Fund: Research Excellence; and the University of Toronto.

\label{sect:bib}
\bibliographystyle{plain}
\bibliography{bib}

\appendix

\section{Formal Definitions}\label{sec:appendix:formal-defs}
We present formal definitions of the notions presented in this paper.
\subsection{Interactive Proofs}\label{sec:appendix:formal-def:IP}
We present the definition of the augmented interactive proof system as in \cite{GoldreichFoC06}. We begin by defining an interactive turing machine (with auxiliary input).

\begin{definition} ``(An Interactive Machine):
\begin{itemize}
    \item  An \textbf{interactive Turing machine} (ITM) is a (deterministic) multi-tape Turing machine. The tapes are a read-only input tape, a read-only random tape, a read-andwrite work tape, a write-only output tape, a pair of communication tapes, and a read-and-write switch tape consisting of a single cell. One communication tape is read-only, and the other is write-only.
    \item Each ITM is associated a single bit $\sigma \in \{0,1\}$, called its \textbf{identity}. An ITM is said to be active, in a configuration, if the content of its switch tape equals the machine’s identity. Otherwise the machine is said to be idle. While being idle, the state of the machine, the locations of its heads on the various tapes, and the contents of the writable tapes of the ITM are not modified.
    \item The content of the input tape is called \textbf{input}, the content of the random tape is called random input, and the content of the output tape at termination is called \textbf{output}. The content written on the write-only communication tape during a (time) period in which the machine is active is called the \textbf{message sent} at that period. Likewise, the content read from the read-only communication tape during an active period is called the \textbf{message received} (at that period).
    
    (Without loss of generality, the machine movements on both communication tapes are in only one direction, e.g., from left to right.)"
\end{itemize}    
\end{definition}
\noindent An interactive machine can be augmented to have an additional read-only tape called the \textbf{auxiliary-input tape}. The content of
this tape is called \textbf{auxiliary input}.

``The complexity of such an interactive machine is still measured as a function of the
(common) input length. Namely, the interactive machine A has time-complexity $t: \mathbb{N} \to \mathbb{N}$ if for every interactive machine $B$ and every string $x$, it holds that when interacting with machine $B$, on common input $x$, machine $A$ always (i.e., regardless of the content of its random tape and its auxiliary-input tape, as well as the content of $B$’s tapes) halts within $t(|x|)$ steps.

We denote by $\langle A(y), B(z)\rangle(x)$ the random variable representing the (local) output
of $B$ when interacting with machine $A$ on common input $x$, when the random input
to each machine is uniformly and independently chosen, and $A$ (resp., $B$) has
auxiliary input $y$ (resp., $z$)."

\begin{definition} \textbf{(Interactive Proof System with auxiliary input):} `` A pair of interactive machines (P,V) is called an interactive proof system for a language $L$ if machine $V$ is polynomial-time and the following two conditions hold:
    \begin{itemize}
        \item Completeness: For every $x \in L$, there exists a string $y$ such that for every $z \in \{0,1\}^*$
        $$\Pr[\langle P(y), V(z)\rangle(x) = 1] \ge \frac{2}{3}$$
        \item Soundness:
        For every $x \notin L$ and every interactive machine $B$, and every $y, z \in \{0,1\}^*$
        $$\Pr[\langle B, V\rangle (x) = 1] \le \frac{1}{3}$$"
    \end{itemize}
\end{definition}
It is to be noted that while the verifier $V$ in the above definition is required to be a (probabilistic) polynomial time machine, the prover is unbounded in its computational power.

\subsection{Proofs of Knowledge}\label{sec:appendix:formal-def:PoK}
We present the definition as in \cite{HL10}, adapted from \cite{BG92}. We assume that both $\mathcal{P}$ and $V$ are PPT machines, and begin by defining an interactive function as in \cite{BG92}. We replace the source's variable name $x$ with $\psi$ following our notation for common input.
\begin{definition}
    An interactive function $A$ associates to each $\psi \in \{0,1\}^*$ and $\eta \in \{0,1\}^*$ (prefix of a conversation) a probability distribution on $\{0, 1\}^*$.
\end{definition}

Let $R \subseteq \{0,1\}^* \times \{0,1\}^*$ be an NP relation. i.e. if $(x, w) \in R$, then $|w| < p(|x|)$ for some polynomial $p$ (The $|\cdot|$ operator denotes the length of the string). Now, we define $L_R = \{x : (x, w) \in R\}$ as in \cite{HL10}. A proof of knowledge for the relation $R$ is defined as follows.
\begin{definition}
    Let $\kappa:\{0,1\}^* \to [0,1]$ be a function. A protocol $(\mathcal{P},V)$ is a proof of knowledge for the relation $R$ with knowledge error $\kappa$, if it satisfies the following properties:
    \begin{itemize}
        \item \textbf{Completeness.} If $\mathcal{P}$ and $V$ follow the protocol on input $\psi$ and private input $w$ to $\mathcal{P}$ where $(\psi,w ) \in R$, then $V$ always accepts.
        \item \textbf{Knowledge soundness/validity.} There exists a constant $c > 0$ and a probabilistic oracle machine $K$, called the knowledge extractor, such that for every interactive prover function $\mathcal{P}^*$ and every $\psi \in L_R$, the machine $K$ satisfies the following condition. Let $\epsilon(\psi)$ be the probability that $V$ accepts on input $\psi$ after interacting with $\mathcal{P}^*$. If $\epsilon(\psi) > \kappa(\psi)$, then upon input $\psi$ and oracle access to $\mathcal{P}^*$, the machine $K$ outputs a string $w$ such that $(\psi, w) \in R$ within an expected number of steps bounded by
        $$ \frac{|\psi|^c}{\epsilon(\psi) - \kappa(\psi)}$$
    \end{itemize}
\end{definition}

\subsection{Zero Knowledge Proofs}\label{sec:appendix:formal-def:ZKP}
We present the definition of Perfect Zero-Knowledge as in \cite{GoldreichFoC06} augmented to accept auxiliary inputs.
\begin{definition}[Perfect Zero-Knowledge with respect to auxiliary input]
Let $(P,V)$ be an interactive proof system with auxiliary inputs for some language $L$. For $x \in L$, let
$P_L(x)$ be the set of prover auxiliary inputs $y$ satisfying the completeness
condition with respect to $x$. We say that $(P,V)$ is \textbf{perfect
zero-knowledge with respect to auxiliary input} if for every probabilistic
polynomial-time interactive machine $V^*$ there exists a probabilistic
algorithm $M^*$, running in time polynomial in the length of its first input,
such that for every $x \in L$, every $y \in P_L(x)$, and every
$z \in \{0,1\}^*$ the following two conditions hold:
\begin{enumerate}
  \item With probability at most $\tfrac{1}{2}$, on input $(x,z)$, machine
    $M^*$ outputs a special symbol denoted $\bot$
    \big(i.e., $\Pr[M^*(x,z) = \bot] \le \tfrac{1}{2}$\big).

  \item Let $m^*(x,z)$ be a random variable describing the distribution of
    $M^*(x,z)$ conditioned on $M^*(x,z) \neq \bot$
    \big(i.e., $\Pr[m^*(x,z) = \alpha] = \Pr[M^*(x,z) = \alpha \mid M^*(x,z) \neq \bot]$
    for every $\alpha \in \{0,1\}^*$\big). Then the following random variables
    are identically distributed:
    \begin{itemize}
      \item $\langle P(y), V^*(z) \rangle(x)$ \quad (i.e., the output of the
        interactive machine $V^*$, on auxiliary input $z$, after interacting
        with the interactive machine $P$, on auxiliary input $y$, on common
        input $x$)
      \item $m^*(x,z)$ \quad (i.e., the output of machine $M^*$ on input
        $(x,z)$, conditioned on not being $\bot$)
    \end{itemize}
\end{enumerate}
Machine $M^*$ is called a \textbf{perfect simulator} for the interaction of
$V^*$ with $P$.
\end{definition}

\section{Functionalities}\label{sec:appendix:functionalities}
We present the ideal functionalities $\Func[ZK], \Func[Clause]$ and $\Func[FlexZKArray]$ as presented in \cite{LAH+22} in Figure \ref{func:zk}, \ref{func:clause}, and \ref{func:zkarray} to define the $k$-Chain-\textsc{ZkUnsat} protocol (Figure \ref{prot:k-chain-zkunsat}).

\begin{figure}[!t]
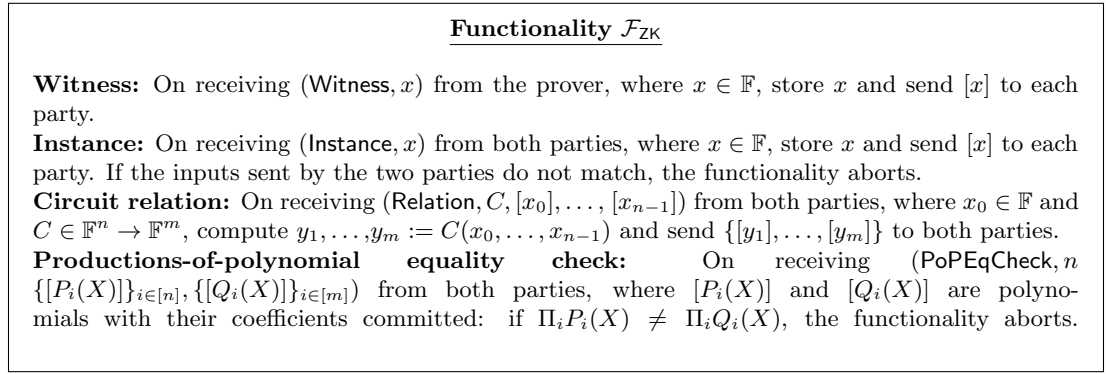

\begin{nffunc}{$\Func[ZK]$}
\textbf{Witness:}
On receiving $({\sf Witness}, x)$ from the prover, where $x\in\mathbb{F}$, store $x$ and
send $[x]$ to each party.

\textbf{Instance:}
On receiving $({\sf Instance}, x)$ from both parties, where $x\in \F$, store $x$ and
send $[x]$ to each party. If the inputs sent by the two parties do not match, the functionality aborts.

\textbf{Circuit relation:}
On receiving $({\sf Relation}, C, [x_0],\ldots,$ $[x_{n-1}])$ from both parties,
where $x_0\in \F$ and $C\in \F^n \rightarrow\F^m$,
compute $y_1,\ldots,$$y_m:=C(x_0,\ldots,x_{n-1})$ and send $\{[y_1],\ldots,[y_m]\}$ to both parties.

\textbf{Productions-of-polynomial equality check:} On receiving $({\sf PoPEqCheck}, n$ 
   $\{[P_i(X)]\}_{i\in[n]},  \{[Q_i(X)]\}_{i\in[m]})$ from both parties, where $[P_i(X)]$ and $[Q_i(X)]$ are polynomials with their coefficients committed: if $\Pi_iP_i(X)\neq \Pi_i Q_i(X)$, the functionality aborts.
\end{nffunc}
\caption{Functionality for zero-knowledge proofs of circuit satisfiability and polynomials.}
    \label{func:zk}
\end{figure}

\begin{figure}[!t]
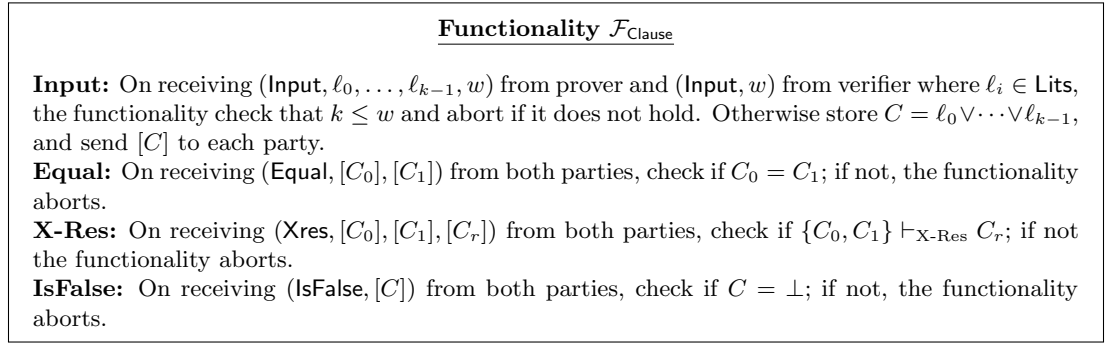

\begin{nffunc}{$\Func[Clause]$}
\textbf{Input:}
On receiving $({\sf Input}, \ell_0, \dots, \ell_{k-1}, w)$ from prover and $({\sf Input}, w)$ from verifier where $\ell_i \in \sf Lits$, the functionality check that $k\leq w$ and abort if it does not hold. Otherwise store $C = \ell_0\lor  \cdots \lor \ell_{k-1}$, and send $[C]$ to each party.

\textbf{Equal:} On receiving $({\sf Equal}, [C_0], [C_1])$ from both parties, check if $C_0 = C_1$; if not, the functionality aborts.

\textbf{X-Res:} On receiving $({\sf Xres}, [C_0], [C_1], [C_r])$ from both parties, check if $\{C_0, C_1\} \vdash_{\text{X-Res}} C_r$; if not the functionality aborts.

\textbf{IsFalse:} On receiving $({\sf IsFalse}, [C])$ from both parties, check if $C = \bot;$ if not, the functionality aborts.
\end{nffunc}
\caption{Functionality for ZK operations on clauses.}
    \label{func:clause}
\end{figure}

\begin{figure}[!t]
\begin{nffunc}{$\Func[FlexZKArray]$}
\textbf{Array initialization:}
On receiving $({\sf Init}, N, [m_0], \dots, [m_{N-1}])$ from $\mathcal{P}$ and $\mathcal{V}$, where $ m_i \in \mathbb{F}$, store the $m_i$ and set $f:={\sf honest}$ and ignore subsequent initialization calls.

\textbf{Array read:}
On receiving $({\sf Read}, \ell, d, t)$ from $\mathcal{P}$, and $({\sf Read}, t)$ from $\mathcal{V}$, where $d \in \F$ and $\ell, t \in \mathbb{N}$, send $[d]$ to each party. If $d \ne m_\ell$ or $t$ from both parties do not match or $\ell \ge t$ then set $f:= {\sf cheating}$.

\textbf{Array check:}
Upon receiving $({\sf check})$ from $\mathcal{V}$ do: If $\mathcal{P}$ sends $({\sf cheat})$ then send ${\sf cheating}$ to $\mathcal{V}$. If $\mathcal{P}$ sends ${\sf continue}$ then send $f$ to $\mathcal{V}$,
\end{nffunc}
\caption{Functionality for weak random access arrays in ZK.}
    \label{func:zkarray}
\end{figure}

\section{The k-Chain-ZKUNSAT Protocol}\label{sec:appendix-protocol}
We present the $k$-Chain-\textsc{ZkUnsat} protocol (Figure \ref{prot:k-chain-zkunsat}) closely modeled after the "CheckProof" protocol in \cite{LAH+22}.

\begin{figure}[!t]
\begin{nfprot}{$k$-Chain-\textsc{ZkUnsat}}
\textbf{Inputs:}
Both parties have formula $\varphi = C_0 \land \dots \land C_{|\varphi|-1}$. Prover $\mathcal{P}$ has a $k$-normalized weakened resolution proof $((h_{10}, h_{11}, \dots, h_{1k}), \dots, (h_{|\pi|0}, \dots, h_{|\pi|k}))$ (This notation denotes the ordered indices of the clauses being resolved in each chain); both parties know the length of the normalized proof $|\pi|$, the chain bound $k$, and the maximum clause width in the proof $d$ (This is the maximum width of all clauses including the intermediate clauses derived via resolution within a chain).

\textbf{Protocol:}

\begin{enumerate}
    \item The two parties obtain $[C_i]_{i \in [0, |\varphi|-1]}$ using $\Func[Clause]$: since $\varphi$ is known to both parties, it uses $\sf instance$ to authenticate the coefficients.
    \item $\mathcal{P}$ locally gets $C_{|\varphi|-1+i}$ for  $i \in [1, |\pi|]$ from their normalized weakened resolution proof. The two parties obtain $[C_i]_{i \in [|\varphi|, |\varphi| - 1+|\pi|]}$ using \Func[Clause] using $\sf witness$ to authenticate the coefficients.
    \item The two parties send $({\sf Init}, |\varphi|+|\pi|, [C_0], \dots, [C_{|\varphi|+|\pi|-1}])$ to \Func[FlexZKArray].
    \item For the $i$-th iteration, the two parties advance the proof check by doing the following for $j$ ranging from 1 to $k$:
    \begin{enumerate}
        \item[(a)] When $j=1$, the prover looks up the tuple $(h_{i0}, h_{i1})$ from the k-normalized weakened resolution proof, and computes the clause $C_{i}^{(1)}$ such that $\{C_{h_{i0}}, C_{h_{i1}}\} \vdash_{\text{X-Res}} C_{i}^{(1)}$ and $C_i^{(1)} \odot C_{h_{i2}} \odot \dots \odot C_{h_{ik}}=C_{|\varphi|-1+i}$. \\
        
        When $j\in \{2, \dots, k\}$, the prover looks up $h_{ij}$ from the k-normalized weakened resolution proof, and computes a clause $C_{i}^{(j)}$ such that $\{C_{i}^{(j-1)}, C_{h_{ij}}\} \vdash_{\text{X-Res}} C_{i}^{(j)}$ and $C_i^{(j)} \odot \dots \odot C_{h_{ik}} = C_{|\varphi|-1+i}$.
        \item[(b)] Fetching the premises: when $j=1$, the prover sends $(\textsf{Read}, h_{i0}, C_{h_{i0}}, |\varphi|-1+i)$ and $(\textsf{Read}, h_{i1}, C_{h_{i1}}, |\varphi|-1+i)$ to $\Func[FlexZKArray]$; $\mathcal{V}$ sends $(\textsf{Read}, |\varphi|-1+i)$ twice to $\Func[FlexZKArray]$, from which the two parties obtain $[C_{h_{i0}}]$ and $[C_{h_{i1}}]$. \\
        
        When $j\in \{2, \dots, k\}$ the prover sends $(\textsf{Read}, h_{ij}, C_{h_{ij}}, |\varphi|-1+i)$ to $\Func[FlexZKArray]$; $\mathcal{V}$ sends $(\textsf{Read}, |\varphi|-1+i)$ to $\Func[FlexZKArray]$, from which the two parties obtain $[C_{h_{ij}}]$. Similarly, the two parties obtain $[C_{|\varphi|-1+i}]$.
        \item[(c)] Committing intermediate clauses: when $j \in \{1, \dots, k-1\}$, the two parties obtain $[C_i^{(j)}]$ using $\Func[Clause]$ with $\sf witness$ to authenticate the coefficients.
        
        \item[(d)]Checking the chain of inferences: When $j = 1$, the two parties send $(\textsf{Xres}, [C_{h_{i0}}], [C_{h_{i1}}], [C_i^{(1)}])$ to $\Func[Clause]$. When $j \in \{2, \dots, k-1\}$, the two parties send $(\textsf{Xres}, [C_i^{(j-1)}], [C_{h_{ij}}], [C_i^{(j)}])$ to $\Func[Clause]$, and when $j = k$, the two parties send $(\textsf{Xres}, [C_i^{(k-1)}], [C_{h_{ik}}], [C_{|\varphi|-1+i}])$ to $\Func[Clause]$.
    \end{enumerate}
    \item After $|\pi|$ iterations, the two parties use $\Func[Clause]$ to check that $[C_{|\varphi|-1+|\pi|}]$ equals $\bot$; if the functionality aborts, $\mathcal{V}$ aborts.

    \item The two parties send $(\textsf{check})$ to $\Func[FlexZKArray]$; if the functionality aborts, $\mathcal{V}$ aborts.
\end{enumerate}

\end{nfprot}
\caption{Protocol for checking a $k$-normalized weakened resolution proof.}
    \label{prot:k-chain-zkunsat}
\end{figure}

\section{Proof of Security}\label{sec:appendix:sec-proof}

We note that Theorem \ref{thm:k-chain-zk} and Proposition \ref{prop:main-paper} are restated in this section under their original numbering.

Our analysis repeatedly needs to recompute and verify weakened resolvents, so we first record a convenient characterization of them.

\begin{fact}\label{fact:wres}
    Let $V$ be the set of variables in a formula $\varphi$, $\mathrm{Lits}=\{x|\neg x:x\in V\}$ and let $\ell\in\mathrm{Lits}$. A clause $C''$ is a \emph{weakened resolvent} of $C$ and $C'$ on $\ell$ iff
    $$(C\setminus\{\ell\})\cup(C'\setminus\{\neg\ell\})\subseteq C''.$$
    In particular, every superset of a weakened resolvent is again a weakened resolvent. The smallest such $C''$ is $(C\setminus\{\ell\})\cup(C'\setminus\{\neg\ell\})$ itself.
\end{fact}

We now make explicit the NP relation that $k$-Chain-\textsc{ZkUnsat} proves knowledge of.

\begin{definition}\label{def:rk}
Fix a chain length $k\ge 1$, and write $\psi=(\varphi,k,|\pi|,d)$ for the common input, with $k$, $|\pi|$, and $d$ encoded in unary, where $\varphi$ is the input formula, $|\pi|$ the number of chains, and $d$ the clause-width bound. A \emph{witness} $\Pi$ for $\psi$ is given as a sequence of clauses $D_0,\dots,D_{|\varphi|+|\pi|-1}$ together with, for each chain $j\in\{1,\dots,|\pi|\}$, premise indices $p_{j,0},\dots,p_{j,k}$ (not necessarily distinct) and pivot literals $\ell_{j,1},\dots,\ell_{j,k}\in\mathrm{Lits}$. Write $C_{j,0}=D_{p_{j,0}}$ and, for $1\le i\le k$, let
\begin{linenomath}
\[ C_{j,i} \;=\; \bigl(C_{j,i-1}\setminus\{\ell_{j,i}\}\bigr)\cup\bigl(D_{p_{j,i}}\setminus\{\neg\ell_{j,i}\}\bigr) \]
\end{linenomath}
be the least weakened resolvent at step $i$ (Fact~\ref{fact:wres}). The witness must satisfy:
\begin{enumerate}
\item[(i)] $D_0,\dots,D_{|\varphi|-1}$ are the clauses of $\varphi$ in the public order fixed by $\psi$;
\item[(ii)] $0\le p_{j,i} < |\varphi|+j-1$ for each $i\in\{0,\dots,k\}$, and the chain result satisfies $C_{j,k}\subseteq D_{|\varphi|+j-1}$;
\item[(iii)] every clause $D_t$ and every intermediate resolvent $C_{j,i}$ ($1\le i\le k-1$) has width at most $d$;
\item[(iv)] $D_{|\varphi| + |\pi| - 1} = \bot$.
\end{enumerate}
Define $R_k = \{(\psi, \Pi) : \Pi \text{ is a witness for } \psi\}$ and $L_{R_k} = \{\psi : \exists\,\Pi\ (\psi,\Pi)\in R_k\}$.
\end{definition}

\begin{lemma}\label{lem:np}
Fix a chain length $k\ge 1$. The relation $R_k$ (Definition~\ref{def:rk}) is an $\mathsf{NP}$ relation.
\end{lemma}
\begin{proof}
Given $\psi$ and a candidate witness $\Pi$, the test reads $\Pi$, recomputes each $C_{j,i}$ by the recurrence of Definition~\ref{def:rk}, and verifies conditions (i) to (iv), in $\mathrm{poly}(|\psi|+|\Pi|)$ time. Since $k$, $|\pi|$, and $d$ are encoded in unary (Definition~\ref{def:rk}), the witness comprises $|\varphi|+|\pi|$ clauses of width at most $d$ together with, per chain, $k+1$ indices and $k$ pivot literals, so $|\Pi|=\mathrm{poly}(|\psi|)$ (a width-$w$ clause encodes as a degree-$w$ polynomial~\cite{LAH+22}). Hence $R_k$ is an $\mathsf{NP}$ relation.
\end{proof}

\begin{lemma}\label{lem:adequacy}
Fix a chain length $k\ge 1$. For every CNF formula $\varphi$, $\varphi$ is unsatisfiable if and only if $(\varphi,k,|\pi|,d)\in L_{R_k}$ for some $|\pi|$ and $d$.
\end{lemma}
\begin{proof}
This follows from the soundness and completeness of weakened resolution~\cite{LAH+22} together with the fact that \textsc{Normalize} (Algorithm~\ref{alg:normalize}) produces a chained proof with chains of length exactly $k$.
\end{proof}

Establishing zero-knowledge requires isolating exactly what each functionality call exposes, so we separate every prover invocation into a public and a secret component, and prove that the order of invocation of functionalities if fixed for a fixed $\psi$.

\begin{definition}[Public and secret components]\label{def:components}
  Consider a single invocation the honest prover sends to a hybrid functionality $\mathcal{F}\in\{\Func[ZK],\Func[Clause],\Func[FlexZKArray]\}$. Its input splits into two components. A \emph{public
  component}: the invoked operation together with its public arguments, and a \emph{secret component}: the private inputs of the prover for the invoked operation.
\end{definition}

For example, in a $\textsf{Read}$ to $\Func[FlexZKArray]$ the prover sends $(\textsf{Read},\ell,d,t)$. The public component is the operation $\textsf{Read}$ together with the bound $t$, and the secret component is the index $\ell$ and the retrieved value $d$.

\begin{definition}[Public invocation pattern]\label{def:backbone}
  Fix a common input $\psi=(\varphi,k,|\pi|,d)$ and a witness $\Pi$ of $\psi$. The honest prover's control flow is determined by $(\psi,\Pi)$, so its invocations to the hybrid
  functionalities occur in a fixed order; let there be $N$ of them. The \emph{public invocation pattern} is
  \[
    B(\psi,\Pi)=\big(\,(\mathcal{F}_i,\mathsf{pub}_i)\,\big)_{i=1}^{N},
  \]
  where $\mathcal{F}_i$ is the target functionality of the $i$th invocation and $\mathsf{pub}_i$ its public component (Definition~\ref{def:components}).
\end{definition}

\begin{lemma}\label{lem:backbone}
Fix a common input $\psi$. The public invocation pattern $B(\psi,\Pi)$ is identical for every witness $\Pi$ of $\psi$.
\end{lemma}

\begin{proof}
A $\textsf{Read}$ takes a public bound and returns a fresh commitment to the value at a prover-chosen secret index $\ell$~\cite[Fig.~5]{LAH+22}. At $\textsf{Init}$ the $|\varphi|$ public clauses of $\varphi$ are fixed as the instance in cells $0,\dots,|\varphi|-1$ of $\Func[FlexZKArray]$~\cite{LAH+22}. Because \textsc{Normalize} makes every chain length exactly $k$, chain $j$ issues $k+1$ premise $\textsf{Read}$s with public bound $|\varphi|+j-1$, makes $k$ $\textsf{X\text{-}Res}$ calls committing each step resolvent through $\Func[Clause]$ with public width $d$, and appends its result to the cell at public index $|\varphi|+j-1$; the run closes with one $\textsf{IsFalse}$ and one $\textsf{check}$. Every public argument is thereby a function of $\psi$ and the public chain index $j\in\{1,\dots,|\pi|\}$, hence of $\psi$ alone.
\end{proof}

\begin{lemma}[Pivot recovery]\label{lem:pivot}
    Let  $(\textsf{X\text{-}Res}, [C],[C'],[C''])$ be a non-aborting call to \Func[Clause]. A pivot literal $\ell\in\mathrm{Lits}$ witnessing this, that is, with $(C\setminus\{\ell\})\cup(C'\setminus\{\neg\ell\})\subseteq C''$, is computable from $C,C',C''$ in time polynomial in $|\psi|$.
\end{lemma}

\begin{proof}
The call was certified on some supplied pivot literal, so a valid $\ell$ exists. By Fact~\ref{fact:wres}, testing $(C\setminus\{\ell\})\cup(C'\setminus\{\neg\ell\})\subseteq C''$ for each of the $2|V|$ literals $\ell$ and returning one that passes recovers a valid pivot literal in $O(|V|)$ subset tests on width-$\le d$ clauses.
\end{proof}

With these components in place, we can state and prove the main security guarantee.

\begin{reptheorem}{thm:k-chain-zk}
Let $R_k$ be the relation in Definition~\ref{def:rk}. Against static corruption, $k$-Chain-\textsc{ZkUnsat} perfectly UC-realizes~\cite{Canetti20} the zero-knowledge proof-of-knowledge functionality $\mathcal{F}^{R_k}_{\mathsf{ZKPoK}}$~\cite{HL10} in \textsc{ZkUnsat}'s $(\Func[ZK],\Func[Clause],\Func[FlexZKArray])$-hybrid model.
\end{reptheorem}

\begin{proof}
Fix an arbitrary environment $\mathcal{Z}$ and, without loss of generality, let the real
adversary be the dummy adversary. For each static corruption pattern we describe a
simulator $\mathcal{S}$ such that the hybrid and ideal execution ensembles are identically
distributed, and hence $\mathcal{Z}$'s views coincide. The both-honest and both-corrupt
patterns are immediate, so it suffices to treat the two interesting cases.\\\\
\textbf{Malicious verifier (zero-knowledge):} The functionality
$\mathcal{F}^{R_k}_{\mathsf{ZKPoK}}$ leaks only the common input
$\psi=(\varphi,k,|\pi|,d)$. By Lemma~\ref{lem:backbone} the public invocation pattern
depends on $\psi$ alone, that is $\mathcal{B}(\psi,\Pi)=\mathcal{B}(\psi)$, so from $\psi$
alone $\mathcal{S}$ reproduces every functionality call, answering each with an opaque
handle and the non-abort verdict that an honest prover with $\Pi\in R_k$ induces. The
handles hide the committed values, so the joint view of $\mathcal{Z}$ and the corrupt
verifier is identical in the two worlds.\\\\
\textbf{Malicious prover (knowledge soundness):} The simulator $\mathcal{S}$ internally
emulates $\mathcal{F}_{\mathsf{ZK}}$, $\mathcal{F}_{\mathsf{Clause}}$, and
$\mathcal{F}_{\mathsf{FlexZKArray}}$ honestly, running their code exactly as specified, and
thereby observes the committed clauses $D_t$, the secret $\mathsf{Read}$ indices $p_{j,i}$,
and a pivot literal $\ell_{j,i}$ for each $\mathsf{X\text{-}Res}$ call, where the pivot is
recovered in polynomial time by Lemma~\ref{lem:pivot}. The non-abort conditions of the
functionalities enforce exactly the four requirements of Definition~\ref{def:rk}:
$\mathsf{Instance}$ enforces (i), the width bound of $\mathsf{Input}$ enforces (iii), the
index bound of $\mathsf{Read}$ together with the $\mathsf{X\text{-}Res}$ checks enforce
(ii), and $\mathsf{IsFalse}$ enforces (iv). Hence the run is non-aborting if and only if
$(\psi,\Pi)\in R_k$, so $\mathcal{S}$ submits a valid witness to
$\mathcal{F}^{R_k}_{\mathsf{ZKPoK}}$ in exactly those runs where the hybrid-world verifier
would accept. Lemma~\ref{lem:np} guarantees that this membership test is well defined, and
Lemma~\ref{lem:adequacy} gives completeness.
\end{proof}

\begin{repproposition}{prop:main-paper}
    Given a parent proof, if the k-Chain-\textsc{ZkUnsat} protocol is executed for all $k \in \{1, \dots, p\}$ where $p$ is the total number of resolution steps in the parent proof, and the k-Chain weakened resolution proofs are derived using the Normalize algorithm, then the collection of all executions reveals the number of chains of length $l$ for every $1 \le l \le p$.
\end{repproposition}
\begin{proof}
    Let $P$ be an arbitrary parent proof, and let $P_k$ for $k \in \{1, \dots, p\}$ be the $k$-Chain weakened resolution proofs derived by the Normalize algorithm. Let there be $m$ chains in $P$, and $L_i$ be the length of the $i^{th}$ chain in $P$.
    
    An adversary who knows the public parameters -- $\left(\psi, k, \sum_i \left\lceil \frac{L_i}{k}\right\rceil, d\right)$ -- for all $k \in \{1, \dots, p\}$ learns $\sum_i \left\lceil \frac{L_i}{k}\right\rceil$ for each $k$. Define $\vec{l} = (l_1, \dots,l_p) \in \mathbb{Z}_{+}^p$ where $l_i = \sum_j \left\lceil \frac{L_j}{i}\right\rceil$. Since $p$ is the total number of resolution steps as revealed by the public parameters in the $k=1$ execution, the adversary learns $l_p=m$. Now, the adversary can define $\vec{t_1} = (t_1^{(1)}=l_1-m, \dots, t_1^{(p)}=l_p-m)$ and infer that the last non-zero entry $t_{1}^{(s-1)}$ denotes the number of chains in $P$ of length $s$. The adversary can now infer the number of lines in each execution that were added as a result of $P$ containing $t_1^{(s-1)}$ chains of length $s$. Let $\vec{e_1} = (e_1^{(1)}, \dots, e_1^{(p)})$ denote the expected addition because of $P$ containing $t_1^{(s-1)}$ chains of length $s$. Define $\vec{t_2} = \vec{t_1} - \vec{e_1}$ and proceed similarly to deduce the number of chains of length $l$ for every $1 \le l \le p$.
\end{proof}

\end{document}